\documentclass[11pt]{amsart}

\usepackage{amsfonts, amsmath, amsfonts, amssymb}
\newtheorem{thm}{Theorem}[section]
\newtheorem{cor}[thm]{Corollary}
\newtheorem{lem}[thm]{Lemma}
\newtheorem{prop}[thm]{Proposition}
\theoremstyle{definition}
\newtheorem{defn}[thm]{Definition}
\theoremstyle{remark}
\newtheorem{rem}[thm]{Remark}
\numberwithin{equation}{section}
\newtheorem{exmp}[thm]{Example}
\numberwithin{equation}{section}
\newtheorem{cntxmp}[thm]{Counterexample}
\numberwithin{equation}{section} 

\numberwithin{equation}{section} \theoremstyle{quest}
 \newtheorem{quest}[]{Question}
 \numberwithin{equation}{section} \theoremstyle{fact}
 
 \numberwithin{equation}{section}
 \theoremstyle{facts}

\numberwithin{equation}{section}

\begin{document}

\title[SamplingMMSp]{A Simple Sampling Method for 
Metric Measure Spaces}

\author{Emil Saucan}

\address{Department of Mathematics, Technion, Haifa, Israel}%
\email{semil@tx.technion.ac.il}%

\thanks{Research supported by 
by European Research Council under the European Community's Seventh Framework Programme
(FP7/2007-2013) / ERC grant agreement n${\rm ^o}$ [203134].}%

\date{\today}

\begin{abstract}
We introduce a new, simple metric method of 
sampling 
metric measure spaces, based on a well-known ``snowflakeing operator''
and we show that, as a consequence of a classical result of Assouad, the sampling of doubling metric spaces is bilipschitz equivalent to that of subsets of some $\mathbb{R}^N$. Moreover, we compare this new method with 
two other approaches, in particular to one that represents a direct application 
of our triangulation method of metric measure spaces satisfying a generalized Ricci curvature condition.
\end{abstract}

\maketitle



\section{Introduction}

Sampling theory has its 
origins
 in the broad field of electrical engineering, more precisely in communication theory, namely in the works of Kotelnikov \cite{Ko}, Nyquist \cite{Ny} and mainly in the seminal papers of Shannon \cite{Sh48}, \cite{Sh}, \cite{Sh-geo}. Naturally, further developments along these lines followed, e.g. \cite{lan3}, \cite{ML}, \cite{SZ2}, as well as natural applications to signal and image processing (see  \cite{U} for an extensive overview\footnote{The importance of sampling theory in the above mentioned fields is also emphasized by the existence of a journal dedicated to this subject: {\it Sampl. Theory Signal Image Process.}}), and also to the related field of graphics (e.g. \cite{CDR}, \cite{DZM}, \cite{LL}). 
Other applications include (but are not restricted to) information theory \cite{Za} and  learning \cite{SZ2}. However, sampling theory rapidly transcended the boundaries of these applicative fields, to become a field of separate and sustained 
interest -- see, e.g. \cite{lan2}, \cite{SZ1}, \cite{Pes}, \cite{Pes1}, amongst others.

Recently, driven mainly by what has become by now a common practice amongst the image processing community, namely to regard images as Riemannian manifolds embedded in higher dimensional spaces (usually $\mathbb{R}^n$ or $\mathbb{S}^n$) -- see e.g. \cite{Ha}, \cite{SZ}, \cite{SL}, \cite{KMS}, \cite{SAZ0} -- a geometric approach to sampling and its implications has emerged \cite{KM}, \cite{SAZ}, \cite{SAZ1}, \cite{Sa1}.

A quite recent theoretical development concerns manifolds with densities \cite{Mo}, \cite{CHHWSX}, and even the more general metric measure spaces \cite{Gr-carte}, \cite{LV}, \cite{St}.
Such objects are not of pure mathematical interest, they arise naturally in a number of applicative fields. 
Amongst these applications, we first mention what is, perhaps, the the most immediate one, namely that of imaging. Indeed, classical grayscale (so called ``natural'') images can be viewed as densities over, say, the unit square. ``Weighted'' manifolds arise naturally in medical imaging, for the density of many types of MRI images is equal to the very proton density. Therefore, a correct modeling of such images, starting from the basic steps of sampling and reconstruction, will arguably produce more accurate results than the ones obtained with present methods. Manifold with densities are also encountered in information geometry \cite{AN}, manifold learning, pattern recognition  (both of these in conjunction with imaging) \cite{OW}, \cite{NajAh}, \cite{PS}, bioinformatics \cite{SA}, graphics \cite{KS} and, perhaps most naturally, communication networks \cite{Sa2}. (In this last case, node and edge weights represent the relevant, characteristic features of the network.)

The 
purpose of this article is rather straightforward: To introduce a simple 
sampling method for metric measure spaces and to compare it to other approaches.
Fittingly,
the structure of the paper is also quite simple: In the following section we discuss 
 two other sampling methods, the accent being placed upon the one stemming from our previous curvature-based triangulation method \cite{Sa} of metric measure spaces satisfying a generalized Ricci curvature condition, as introduced by Lott and Villani \cite{LV} and Sturm \cite{St}. In particular, we show that, indeed, the vertices of the said triangulation represent a sampling of the given space. Moreover, 
based on the construction employed in \cite{Sa}, 
we also bring 
a result regarding the topological dimension of weak $CD(K,N)$ spaces.
 Section 3 represents the heart of paper: We introduce here the new sampling method, based upon of a well known ``snowflaking operator'', and we show, as a consequence of a classical result of Assouad \cite{As1}, \cite{As2}, that sampling of doubling metric spaces is equivalent to that of subsets of $\mathbb{R}^N$, for some $N$. We also prove that compact weak $CD(K,N)$ spaces are Ahlfors $N$-regular, for $K \geq 0$. Section 4 represents the final section and is dedicated to a concluding discussion and comparison of the considered sampling methods. For the readers' convenience and for the sake of the paper's self containment, we bring some background material regarding curvature of metric measure spaces in an appendix. 

 As it is clear from the very title and from the brief overview of the history of the subject given above, this paper is motivated not the least by the applicative goals, as they emerge in such fields as information geometry, image processing (and in particular, medical imaging), manifold learning, etc. Therefore, it is only natural that, while this paper is mathematical, not least in the origins of the suggested methods, we shall mention occasionally, in the relevant context, but mainly in the last section, a number of suggested applications.


\section{Other Approaches}


We present below two other methods of sampling metric measure spaces, besides the simple one announced in the title. We do not pretend that we  thus exhaust all the possible approaches to the sampling problem on such manifolds, but rather concentrate on those methods that either are familiar to us or we deem 
to be more natural and/or important.

\subsection{Generalized Ricci Curvature-Based Sampling} \label{section:Ricci}

Since the sampling method based upon the generalized curvature condition of Lott, Villani and Sturm does not represent the one of main interest here, and since the technical definitions involved are rather lengthy and involved, we do not dwell upon them here, so not to disrupt the cursiveness of the exposition. However, to preserve the self-containment of the paper, we adopt a compromise between fluency and clarity,  
by bringing them in an appendix.

Considering that the exposition in \cite{Sa} was rather roundabout and somewhat didactic, and for the self containment of the paper, we bring here a concise proof. For this, we shall need the following notions:

\begin{defn}
Let $(X,d)$ be a metric space and let $p_1,\ldots,p_{n_0}$ be points  $\in X$,  satisfying  the
following conditions:
\begin{enumerate}
\item The set $\{p_1,\ldots,p_{n_0}\}$ is an $\varepsilon$-net on $X$, i.e. the
balls $\beta^n(p_k,\varepsilon)$, $k=1,\ldots,n_0$ cover $X$;
\item The balls 
$\beta^n(p_k,\varepsilon/2)$ are pairwise
disjoint.
\end{enumerate}
Then the set $\{p_1,\ldots,p_{n_0}\}$ is called a {\it minimal
$\varepsilon$-net} and the packing with the balls
$\beta^n(p_k,\varepsilon/2)$, $k=1,\ldots,n_0$, is called an {\it
efficient packing}. The set $\{(k,l)\,|\,k,l = 1,\ldots,n_0\; {\rm
and}\; \beta^n(p_k,\varepsilon) \cap \beta^n(p_l,\varepsilon) \neq
\emptyset\}$ is called the {\it intersection pattern} of the minimal
$\varepsilon$-net (of the efficient packing).
\end{defn}

We our proof begin with the following lemmas:

\begin{lem} \label{lem:2.2++}
Let $(X,d,\nu)$ be a compact weak ${\rm CD}(K,N)$ space, $N < \infty$, such that ${\rm Supp}\nu = X$ and such that ${\rm diam}X \leq D$.
Then there exists $n_1 = n_1(K,N,D)$, such that if $\{p_1,\ldots,p_{n_0}\}$ is a minimal $\varepsilon$-net in $X$, then $n_0 \leq n_1$.
\end{lem}

\begin{rem}
Note that, since $N < \infty$, the condition ${\rm Supp}\nu = X$ imposes no real restriction on $X$ (see \cite{Vi}, Theorem 30.2 and Remark 30.3).
\end{rem}

\begin{lem} \label{lem:2.3++}
Let $(X,d,\nu)$ be a compact weak ${\rm CD}(K,N)$ space, $N < \infty$, such that ${\rm Supp}\nu = X$ and such that ${\rm diam}X \leq D$.
Then there exists $n_2 = n_2(N,K,D)$, such that, for any $x \in M^n$,
$|\{j \,|\, j = 1,\ldots,n_0\; {\rm and}\;
\beta^n(x,\varepsilon) \cap \beta^n(p_j,\varepsilon)$ $ \neq
\emptyset\}| \leq n_2$, for any minimal $\varepsilon$-net
$\{p_1,\ldots,p_{n_0}\}$.
\end{lem}

It is important to note that the integer $n_2$ in the lemma above does not depend upon $\varepsilon$.

\begin{lem} \label{lem:2.4++}
Let $(X_1,d_1,\nu_1)$ and $(X_2,d_2,\nu_2)$ be as in Lemma \ref{lem:2.2++}.
and let $\{p_1,\ldots,p_{n_0}\}$ and $\{q_1,\ldots,q_{n_0}\}$ be minimal $\varepsilon$-nets with the same
intersection pattern, on $X_1$, $X_2$, respectively. Then
there exists a constant $n_3 = n_3(N,K,D,C)$, such that if
$d_1(p_i,p_j) < C\cdot\varepsilon$, then $d_2(q_i,q_j) <
n_3\cdot\varepsilon$.
\end{lem}

The purpose 
of the lemmas above is to allow the construction of a triangulation of a weak (compact) metric measure space, as follows. Construct a simplicial complex having as vertices the centers of the balls $\beta^n(p_k,\varepsilon)$, in the following manner:  Edges are connecting the centers of adjacent balls; further edges being added to ensure the cell complex obtained is triangulated to obtain a simplicial complex. (A concise and 
elegant exposition of these ideas in the classical (geometric differential) context, can be found in \cite{Br}.)
The construction 
is possible since, by definition, weak ${\rm CD}(K,N)$ spaces are geodesic. Moreover, in {\it nonbranching} spaces, the geodesics connecting two vertices of the triangulation are unique a.e. (see, e.g. \cite{Vi}, Theorem 30.17). (Recall that a geodesic metric  space $X$ is called nonbranching iff any two geodesics $\gamma_1, \gamma_2: [0,t] \rightarrow X$ that coincide on a subinterval $[0,t_0], 0 < t_0 < t$, coincide on $[0,t]$.)

\begin{rem}
In {\it smooth metric measure spaces} (see the Appendix) one can actually produce a convex triangulation by choosing $\varepsilon$ to equal the {\it convexity radius} ${\rm ConvRad}(M^n) = \inf\{r>0\,|\, \beta^n(x,r) \; {\rm is\;
convex},\; {\rm for\; all\;} x \in M^n\}$, and, moreover, control the convexity radius via the {\it injectivity radius} ${\rm InjRad}(M^n) =
\inf\{{\rm Inj}(x)\,|\,x \in M^n\}$, where
${\rm Inj}(x) =  \sup{\{r\,|\, {\rm exp}_{x}|_{\mathbb{B}^n(x,r)}\; {\rm is\; a\;} }$ ${\rm diffeomorphism}\}$.
What renders the construction above into a simple and practical triangulation method is a classical result of Cheeger, similar, later ones (see, e.g. \cite{Be}) showing that, in this case, there exists a universal positive lower bound for ${\rm InjRad}(M)$ in terms of
$k, D$ and $v$, where $v$ is the lower bound for the volume of $M$.

However, many weak ${\rm CD}(K,N)$ spaces of interest fail to be locally convex, rendering the construction above impossible. 
Fortunately, local convexity does hold for an important class of metric measure spaces: Indeed, by \cite{Ptr}, \cite{ZZ1}, 
${\rm Alex}[K] \subset {\rm CD}((m-1)K,m)$, where ${\rm Alex}[K]$ denotes the class of $m$-dimensional {\it Alexandrov spaces} with curvature $\geq K$ (see \cite{BBI}, \cite{Gr-carte}), equipped with the volume measure.
\end{rem}

We can formalize the construction above as:

\begin{thm}
$(X,d,\nu)$ be a compact weak ${\rm CD}(K,N)$ space. Then $X$ is triangulable. Moreover, if $X$ is locally convex, in particular if it is a $m$-dimensional {\it Alexandrov spaces} with curvature $\geq K$, equipped with the volume measure, one can ensure that the simplices of the triangulation are convex. 
\end{thm}

\begin{rem}\label{thm:triang}
For further improvements of the triangulation above and their applications, see \cite{Sa}. 
\end{rem}

Of course, one is still has to ask himself if this triangulation result also represents, indeed, a sampling method?

\begin{itemize}
\item It is quite classical by now that it represents a sampling method at least in the {\em metric} sense. More precisely, one assuredly obtains  the Gromov-Hausdorff convergence of the  $\varepsilon${\it-net} (viewed as a metric space) to $M$ when $\varepsilon \rightarrow 0$ (see \cite{Gr-carte}, \cite{BBI}).

\item Moreover, the same $\varepsilon$-net also converges in {\em measure}, for instance in the {\it measured} Gromov-Hausdorff topology (see, e.g. \cite{Vi}). For the choice of the weights in this context, see the next item. For a more detailed discussion regarding sampling based upon measures see the following subsection.

\item However, given the fact that the vertices of the triangulation (points of the $\varepsilon$-net) were constructed using generalized Ricci curvature, and since for the types of convergence above one hardly needs such sophisticated tools, one is conducted to ask whether our method also assures convergence in the {\em curvatures} sense, i.e. that it is something akin -- at least in spirit -- to more classical results for $PL$ approximations of manifolds (such as those of \cite{CMS}) or of the Gromov-Hausdorff convergence for Alexandrov spaces (see \cite{Gr-carte}, \cite{BBI}).
    \\
    The answer proves to be positive in this case, too, up to some (technical) required adaptation of the notion of ${\rm CD}(K,N)$ spaces to the ``discrete'' case of graphs (see \cite{BS}). First, let us define precisely the weights (i.e. the measure) for the considered space: Since the balls $\beta^n(p_k,\varepsilon)$ cover $M^n$, any sequence $\mathcal{P}(\varepsilon_m)$, of efficient packings such that $\varepsilon_m \rightarrow 0$ when $m \rightarrow \infty$, generates a {\it discretization} $(X_m,d,\nu_m)$, in the sense of \cite{BS}: Consider the Dirichlet (Voronoi) cell complex (tesselation) $\mathfrak{C}_m = \{C(p_{m,k})_k\}$  of centers $p_{k,m}$ and atomic masses $\nu_m(p_{m,k}) = \nu[C(p_{m,k})]$. Then taking $X_m = \{p_{m,k}\}$, $d$ the original metric of $X$ and $\nu_m$ as defined, provides us with the said discretization.
    It follows, by \cite{BS}, Theorem 4.1, that if ${\rm Vol}(M^n) < \infty$, the sequence $(X_m,d,\nu_m)$ converges in the $W_2$ metric (see \cite{St}, \cite{Vi}) to a metric measure space and, moreover, if $(X,d,\nu)$ is a weak ${\rm CD}(K,N)$ space, then, for small enough $\varepsilon$, so will be $(X_m,d,\nu_m)$,  but only in a generalized (``{\it rough}'') sense.
    It should be noted that, by \cite{BS}, Theorem 3.10, the converse result also holds. More precisely, a (not necessarily compact) metric measure space $(X,d,\nu)$ is a weak ${\rm CD}(K,N)$ space if there exists a family $\{(X_\iota,d_\iota,\nu_\iota)\}_\iota$ of weak ``discrete'' metric measure ${\rm CD}(K_\iota,N_\iota)$ spaces, that converges to $(X,d,\nu)$ (in the Wasserstein metric), where $K_\iota$ is in the weak (or ``rough'') sense, and ${\rm diam} X_\iota \leq D_0$, for some positive $D_0$, and such that $K_\iota \rightarrow K$, when $\iota \rightarrow 0$.
    (For more details regarding the precise definition of rough curvature bounds and the proof of this and other related results, see \cite{BS}.)

\item  The somewhat opposite problem is also of interest in  some applications, and in particular in image processing and graphics. Namely, if $X$ is a topological manifold, such that each of the coordinate patches satisfies a weak ${\rm CD}(K_\iota,N)$ condition, $K_\iota \geq K_0$, one can triangulate each of these coordinate patches, and ``glue'' these local triangulations to obtain a global triangulation. The question is, naturally, whether the resulting triangulation will be a ${\rm CD}(K_0,N)$ space? The answer is positive, due to \cite{BaS}, Theorem 5.1, that  shows that the local ${\rm CD}(K,N)$ condition implies a global, albeit slightly weaker one, and to \cite{Hi}, Theorem 4, where for a proof of the local-to-global property of the ${\rm CD}(K,N)$ condition is given. (A local-to-global property is also proved in \cite{BaS} for a so called {\it reduced} ${\rm CD}(K,N)$ condition.)


\end{itemize}


Before we bring the proofs of the lemmas, let us note that,
incidentally, from Lemma \ref{lem:2.3++} we obtain 
the following result\footnote{that did not appear in \cite{Sa}} 
(compare with \cite{Vi}, Corollary 30.14.):

\begin{cor} \label{lem:top-dim}
$(X,d,\nu)$ be a compact weak ${\rm CD}(K,N)$ space. Then $X$ has topological dimension $\leq n_2$, where $n_2 = n_2(N,K,D)$ is as in Lemma \ref{lem:2.3++}  above.
\end{cor}

\begin{proof}
The corollary 
follows immediately from the following characterization 
of topological dimension  -- see \cite{HW}, Corollary to Theorem V 8. (p. 67) -- namely that
%
a compact space has dimension $\leq n$ iff it has coverings of arbitrarily small mesh and order $\leq n$.
%
The definition (see, e.g. \cite{HW}, Definition V 1.) of the {\it order} of a covering requires precisely the property guaranteed by Lemma \ref{lem:2.3++}.
To show that compact weak ${\rm CD}(K,N)$ spaces enjoy the second required property, i.e. existence of coverings with arbitrarily small meshes, one can proceed, amongst other possibilities, as follows:
Start with the triangulation provided by Theorem \ref{thm:triang} above. Since $X$ is geodesic, one can built, iteratively, the barycentric\footnote{Since the ``thickness'' of the triangles (see, e.g. \cite{Sa} for the definition) is not an issue here, the simple to produce barycentric subdivision will suffice.} subdivisions of any 
order, thus ensuring that the mesh of the triangulations tends to zero.
Consider the balls 
having as diameters the sides of the obtained triangulation. Then, by a classical argument (see, for instance \cite{Be}, \cite{Br}), 
this collection of balls represent, indeed, the required 
covering.
\end{proof}

The main (in fact, the only essential) tool in proving the Lemmas 2.1-2.3 above is, by straightforward analogy with the classical case (see \cite{GP}) the following generalized version of the Bishop-Gromov Comparison Theorem:

\begin{thm}[Bishop-Gromov Inequality for Metric Measure Spaces, \cite{St}] \label{thm:BG++}
Let $(X,d, \nu)$ be a weak ${\rm CD}(K,N)$ space, $N < \infty$, and let $x_0 \in {\rm Supp}\,\nu$.
Then, for any $r > 0$, $\nu(B(x_0,r)) = \nu(B[x_0,r])$. Moreover,
\begin{equation}
\frac{\nu(B[x_0,r])}{\int_0^r{S_K^N(t)}dt}
\end{equation}
is a nonincreasing function of $r$, where
\begin{equation}
S_K^N(t) = \left\{
                  \begin{array}{ll}
                    \Big(\sin{\sqrt{\frac{K}{N-1}}t}\Big)^{N-1} & \mbox{if $K > 0$}\\\\
                    t^{N-1} & \mbox{if $K = 0$}\\\\
                     \Big(\sinh{\sqrt{\frac{|K|}{N-1}}t}\Big)^{N-1} & \mbox{if $K < 0$}
                  \end{array}
           \right.
\end{equation}
(Here, $B(x_0,r), B[x_0,r]$ represent the standard convention for the open, respective closed ball center $x_0$ and radius $r$.)
\end{thm}

\begin{proof}[Proof of Lemma 2.1]
Let $\{p_1,\ldots,p_{n_0}\}$ be a minimal $\varepsilon$-net on $M^n$ and let $\tilde{p}$ be a point in $\widetilde{M}^n_k$ -- the $k$-space form. Then, by Theorem 2.7 above

\[\frac{\nu(B(p,r))}{\nu(B(p,R))} \geq \frac{{\rm Vol}B(\tilde{p},r)}{{\rm Vol}B(\tilde{p},R)}\,, 0 < r < R\,;\]
for any $p \in M^n$.

Let $i_0$ such that $\nu(B(p_{i_0},\varepsilon/2))$ is minimal. By 2.1.(2) it follows that

\[n_0 \leq \frac{\nu(M^n)}{\nu(B(p_{i_0},\varepsilon/2))} \leq \frac{{\rm Vol}B(\tilde{p},D)}{{\rm Vol}B(\tilde{p},\varepsilon/2)}\,.\]
(To obtain the last inequality, just take, in Bishop-Gromov Theorem, $R = {\rm diam}M^n \leq D$.)

The desired conclusion now follows by taking

\[n_1 = \left[\frac{{\rm Vol}B(\tilde{p},D)}{{\rm Vol}B(\tilde{p},\varepsilon/2)}\right]\,.\]
\end{proof}

\begin{proof}[Proof of Lemma 2.3]
Let $j_1,...,j_s$ be such that $B(x,\varepsilon) \cap B(p_{j_i},\varepsilon) \neq \emptyset$. Then $B(p_{j_i},\varepsilon/2) \subset B(x,5\varepsilon/2)$.

Let $k \in \{1,..., s\}$ be such that $B(p_{j_k},\varepsilon/2)$ has minimal measure. Then (as in the proof of Lemma 2.2) it follows that:

\[s \leq \frac{\nu(B(x,5\varepsilon/2))}{\nu(B(p_k,\varepsilon/2))} \leq \frac{\nu(B(p_{j_k},9\varepsilon/2))}{\nu(B(p_{j_k},\varepsilon/2))} \leq \frac{{\rm Vol}B(\tilde{p},9\varepsilon/2)}{{\rm Vol}B(\tilde{p},\varepsilon/2)}\,,\]
where $\tilde{p}$ is as in the proof of the previous lemma. But 

\[\frac{{\rm Vol}B(\tilde{p},9\varepsilon/2)}{{\rm Vol}B(\tilde{p},\varepsilon/2)} = \frac{\int_{0}^{9\varepsilon/2}S_K^n(r)dr}{\int_{0}^{\varepsilon/2}S_K^n(r)dr}\,,\]
and the function
\[h(\varepsilon) = \frac{\int_{0}^{9\varepsilon/2}S_K^n(r)dr}{\int_{0}^{\varepsilon/2}S_K^n(r)dr}\]
extends to a continuous function $\tilde{h}:[0,D] \rightarrow \mathbb{R}_+$,
(since $h(\varepsilon)  \rightarrow 0$ when $\varepsilon \rightarrow 0$).

\end{proof}

\begin{rem}
It is important to note, both for the basic construction, and not least, for Corollary \ref{lem:top-dim}, that $n_2$ is independent of $\varepsilon$.
\end{rem}

\begin{proof}[Proof of Lemma 2.4]
 Evidently, since $d(p_i,p_j) < C \cdot \varepsilon$, it follows that $p_j \in B(p_i,C \cdot \varepsilon)$. Thus, precisely as in the proof of the previous lemma, it follows that there exists $n' = n'(C)$,

\begin{displaymath}
n'(C) = \max{\frac{\int_{0}^{(4k+1)\varepsilon/2}S_K^n(r)dr}{\int_{0}^{\varepsilon/2}S_K^n(r)dr}}\,,
\end{displaymath}
such that at most $n'$ of the balls $B(p_1,\varepsilon/2),...,B(p_{n'},\varepsilon/2)$ are included in $B(p_i,(C+\frac{1}{2})\varepsilon/2$.

Since $\{p_1,\ldots,p_{n_0}\}$ and
$\{q_1,\ldots,q_{n_0}\}$ have the same intersection pattern, it follows that $d(q_i,q_j) \leq n_3(C)$, where $n_3(C) = 2[n'(C)-1]$.

\end{proof}


\subsection{Gromov-Prohorov Metric-Based Sampling (``Measure Decides'')}

The approach above is natural, perhaps, from the point of view of a differential geometer since, as already mentioned, it extends to more general spaces and notions of curvature established ideas and techniques of classical (``proper'') differential geometry. It is however, an involved approach, and certainly not one that will appeal to people working in probability theory, statistic (or even main stream image processing, where histograms are a main ``staple''). Such researchers will ask themselves, whether it is not possible to sample a metric measure space using solely the measure. The answer is, positive, as we shall see below, with the proviso that the metric should somehow be involved, since the goal is to find a sampling method for {\em metric measure} spaces. Therefore, the proper question would be if there exists such a method in which the main role is played by the measure (i.e. ``measure decides'').

We first have to define a ``good'' distance between metric measure spaces, where by ``good'' we mean here that it satisfies the requirement discussed above. Such a distance, would be, for instance, the {\it Gromov-Prokhorov} distance $d_{GP}$: Given two metric measure spaces $\mathcal{X} = (X,d,\mu)$ and $\mathcal{Y} = (Y,\rho,\nu)$, we define
\begin{equation} \label{eq:GromovProkhDist}
d_{GP}(\mathcal{X},\mathcal{Y}) = \inf{d_P(\mu',\nu')}\,.
\end{equation}
Here the infimum is taken over all the measure preserving isometric embeddings $f:(X,d,\mu) \rightarrow (Z,\delta,\lambda)$,  $g:(Y,\rho,\nu) \rightarrow (Z,\delta,\lambda)$, where $(Z,\delta,\lambda)$ is a common metric measure space (so the filiation of the Gromov-Prokhorov distance from the by now classical {\it Gromov-Hausdorff} metric -- see below -- is evident);
and where $d_P$ denotes the classical {\it Prokhorov} distance, that can be defined in a very geometric manner (in the sense that is a straight forward generalization of the {\it Haussdorf} metric -- see (\ref{eq:HausDist}) below) as
\begin{equation}
d_P(\mu',\nu') = \inf\{r > 0\,|\, \mu'(F) \leq \nu'(\mathcal{N}(F)) + r\,,  \nu'(F) \leq \mu'(\mathcal{N}(F)) +r\,, \forall F = \bar{F} \subseteq Z\}\,,
\end{equation}
where $\mathcal{N}(F)$ denotes the $r$-neighbourhood of $F$.

This approach corresponds, according to \cite{Vi}, to the ``mainly measure'' definition of isometry of metric measure spaces, namely that $\varphi:(X,d,\mu) \rightarrow (Y,\rho,\nu)$ is an {\it isometry of metric measure spaces} iff it is a measure preserving isometry of the metric underlying spaces $({\rm Supp}\mu,d|_{\rm Supp\mu})$ and  $({\rm Supp}\nu,\rho|_{\rm Supp\nu})$. Note that the metric indeed plays a role -- albeit ``subdued'' -- since the infimum is taken solely over isometric embeddings (and, moreover the metric does play a role in the ``background'' definition  of the Prokhorov distance).

Again, according to \cite{Vi}, the ``metric and measure, (but mainly metric)'' approach is embodied in the {\it Gromov-Hausdorff-Prokhorov} distance $d_{GHP}$, where
\begin{equation} \label{eq:GromovHausProkhDist}
d_{GHP}(\mathcal{X},\mathcal{Y}) = \inf\{d_H(X,Y) + d_P(\mathcal{X},\mathcal{Y})\}\,,
\end{equation}
where $d_H$ denotes the {\it Hausdorff} distance
\begin{equation} \label{eq:HausDist}
d_H(A,B) = \inf\{r > 0\,|\, A \subseteq B\; {\rm and}\; B \subseteq A\}\,,
\end{equation}
and all the other notations are as above, the infimum in (\ref{eq:GromovHausProkhDist}) being taken as in the definition of the Gromov-Prokhorov distance.

In this setting, the fitting notion of isometry of metric measure spaces is that of {\it isomorphism}  (or {\it measure-preserving isometry}) of metric measure spaces, meaning that, given $(X,d,\mu)$ and $(Y,\rho,\nu)$, there exists a measurable bijection $\psi:\mathcal{X} \rightarrow \mathcal{Y}$ such that (a) it is an isometry between 
$(X,d)$ and $(Y,\rho$); and (b) it preserves measure, that is $\psi_\sharp\mu = \nu$, where, as usual, $\psi_\sharp\mu$ denotes the push-forward of the measure $\mu$, i.e. $\psi_\sharp\mu(B) = \mu\left(\psi^{-1}(B)\right)$, for any Borel set $B \subseteq Y$.\footnote{For other possible distances between metric measure spaces, see \cite{Vi}, pp. 770-771.}

It turns out (see, for instance  \cite{Vi}, Theorem 27.26),
that at least for the significant class of {\it doubling} spaces (see Definition \ref{def:doubling} below), the two approaches are equivalent, and thus interchangeable in any mathematical or practical application.


\section{``Snowflaking  Operator''-Based Sampling 
(``Metric Decides'')}

However simple and alluring the probabilistic approach may appear, to the geometer 
it seems somewhat unnatural. It is even less palatable to those whose interest is drove mainly by possible implementations, e.g. people working in information geometry, image processing, manifold learning, etc.

Therefore, it is a 
natural desire  to find a new metric that 
encapsulates the behaviors of both the original metric and of the given measure, at least as far as sampling (via $\varepsilon$-nets) is concerned.

\subsection{Background: Quasimetrics and Doubling Measures}

As general bibliographical references for the material in this subsection, including missing proofs, we have used \cite{He}, \cite{Se}, \cite{Se1}.

\subsubsection{Quasimetrics}
We begin with the following basic definition:

\begin{defn}
Let $X$ be a nonempty set. $q:X\times X \rightarrow \mathbb{R}_+$ is called a $K$-{\em quasimetric} iff
\begin{enumerate}
\item $q(x,y) = 0$ iff $x=0$;
\item $q(x,y) = q(y,x)$, for any $x,y \in X$;
\item $q(x,y) \leq K(q(x,z)+q(z,y)), {\rm for\; any} \; x,y,z \in X$\,.
\end{enumerate}
\end{defn}

\begin{rem}
Some authors replace condition (2) above by the following weaker one: There exists $C_0 \geq 1$ such that $q(x,y) \leq C_0q(y,x)$, for any $x,y \in X$.
\end{rem}

\begin{rem}
A number of brief comments:
\end{rem}
\begin{itemize}

\item A quasimetric is not necessarily a metric (while obviously, any metric is a quasimetric with $K = 1$).

\begin{cntxmp} \label{exmpl:nonmetric}
The following counterexample is not only the basic one, it is very important to us in the sequel:
\begin{equation} \label{eq:snflkng}
q_s(x,y) = \left(d(x,y)\right)^s
\end{equation} 
is a quasimetric for any $s > 0$, but not, in general, a metric, for $s > 1$.
\end{cntxmp}

\item Quasimetric balls can be defined precisely like metric balls, and the constitute the basis for a topology on $X$.


\item For the next remark we need a definition that may appear a bit superfluous at this point, but it will prove to be highly relevant later on:

\begin{defn}
Let $(X,q)$ and $(Y,\rho)$ be quasimetric spaces, and let $f:X \rightarrow Y$ be an injection. $f$ is called $\eta$-{\em quasisymmetric}, where $\eta:[0,\infty) \rightarrow [0,\infty)$ is a homeomorphism iff
\begin{equation} \label{eq:qs}
\frac{\rho(f(x),f(a))}{\rho(f(x),f(b))} \leq \eta\left(\frac{q(x,a)}{q(x,b)}\right),
\end{equation}
for any distinct points $x,a,b \in X$.
\end{defn}

Intuitively, while quasisymmetric mappings may change the size of balls quite dramatically, they do not change very much their shape.
This fact is important in the next proposition (see, e.g. \cite{Se}), that shows that whereas, as we noted above, $q_s$ is not a metric, the canonical injection $(X,d) \hookrightarrow (X,q_s)$ is quasisymmetric.

\begin{prop} \label{prop:qm-equiv-metric}
Let $q$ be a $K$-quasimetric on $X$. Then, there exists $s_0 = s_0(K)$ such that, for any $0 < s \leq s_0$ there exists a metric $d_s$ on $X$, and a constant $C = C(s,K) \geq 1$, such that
\begin{equation}
\frac{1}{C}q_s(x,y) \leq d_s(x,y) \leq Cq_s(x,y)\,,
\end{equation}
where $q_s$ is as in (\ref{eq:snflkng}), 
i.e. $q_s(x,y) = \left(q(x,y)\right)^s$\,.

\end{prop}

\begin{rem} \label{rem:qm-equiv-metric}
If $q$ is a $K$-quasimetric ($K \geq 1$), then $q_s$ is bilipschitz equivalent to $d_s$, for any $s > 0$, such that $(2K)^{2s} \leq 2$, that is for any $s > 0$ such that
\[s \leq \frac{1}{2}(\log_2{K}+1)\,.\]

Moreover, the bilipschitz constant can be chosen to be
\[C = (2K)^{2s}\,.\]

\end{rem}

\end{itemize}

\subsubsection{From doubling measures to quasimetrics}

We first remind the reader the following basic definition:

\begin{defn} \label{def:doubling}
Let $(X,d,\mu)$ be a metric measure space
$X$ is called {\it doubling} iff $\mu$ is doubling, i.e. iff there exists a constant $D$ such that, for any $x \in X$ and any $r > 0$,
\begin{equation}
\mu\left(B_d[x,2r]\right) \leq D \mu\left(B_d[x,r]\right)\,.
\end{equation}
(Here $B_d[x,r]$ denotes -- as it standardly does -- the closed ball of radius $r$, in the metric $d$.)
A metric measure space $(X,d,\mu)$, where $\mu$ is doubling is sometimes called {\it of homogenuous type}.
\end{defn}

For the record, 
a {\it metric measure space} is a triple $\mathcal{X} = (X,d,\mu)$ where $(X,d)$ is a metric space and $\mu$ is a Borel measure on $X$.

\begin{rem} \label{rem:atoms1}
If $(X,d,\mu)$ is doubling, then it admits {\it atoms} only at isolated points.
\end{rem}

\begin{rem}
The connection between this definition and Section \ref{section:Ricci}, beyond the basic goal in both approaches, resides in the fact that any Riemannian manifold of nonnegative Ricci curvature is doubling (with respect to the volume measure) -- see, e.g. 
\cite{Pet}. (Indeed, it may be that this case represents one of the original motivations for studying doubling spaces.)

Moreover, weak $CD(K,N)$ spaces are locally doubling (on their support) -- see \cite{Vi}, Corollary 30.14 and globally doubling if $X$ has bounded diameter. Since, by \cite{Vi}, Theorem 29.9, smooth weak $CD(K,N)$ spaces are $CD(K,N)$, the same assertions are true for smooth metric measure spaces. In fact, a stronger statement holds, since, by \cite{Vi}, Corollary 18.11, smooth metric measure spaces, with $1 < N < \infty$, are globally doubling.

Since this fact is not less important for our purposes, we remind the reader that if $X$ is doubling (with added provisos of being Polish and compact) then, by \cite{Vi}, Proposition 27.26 and Corollary 27.28, both approaches to convergence of metric measure spaces discussed in Section 2.2 are, indeed, equivalent, as far as convergence of $\varepsilon$-nets -- hence sampling -- is concerned.
\end{rem}

For any $s > 0$, we define the quasimetric $q_{\mu,s}$ as
\begin{equation} \label{eq:def-q}
q_{\mu,s}(x,y) = \big(\mu\left(B[x,d(x,y)]\right) + \mu\left(B[y,d(x,y)]\right)\big)^s\;.
\end{equation}
(This can be written in compact form as $q_{\mu,s}(x,y) = \left(\mu(B_{x,y})\right)^s$, where $B_{x,y} = B[x,d(x,y)] \cup B[y,d(x,y)]$.)
%

\begin{exmp}
If $X = \mathbb{R}^n$, with $\mu \equiv Vol_n$, and if $s = 1/n$, then $q_{\mu,s} \equiv {\rm const}\cdot d_{Eucl}$. (In particular, for $n=2$, $q_{\mu,s}=\frac{\sqrt{\pi}}{2}d_{Eucl}$\,.)
\end{exmp}

\begin{rem}
For $X = \mathbb{R}^n$, one can define $q_{\mu,s}(x,y)$ simply by $q_{\mu,s}(x,y) = \left(\mu\left(B[m,\frac{x+y}{2}]\right)\right)^s$, where $m$ denotes the midpoint of the segment $\overline{xy}$. However, in the general case, and in particular for graphs, one has to use the more general expression (\ref{eq:def-q}).
\end{rem}

Note that, if $K$ is the quasimetric constant of $q_{\mu,s}$, then $K = K(\mu,s)$.

Also, by Proposition \ref{prop:qm-equiv-metric}, there exists $s_0 = s_0(\mu) > 0$, such that $q_{\mu,s}$ is bilipschitz equivalent to a metric $d_{\mu,s}$, for any $0 < s \leq s_0$. This fact will play a crucial role in the remainder of the paper.

\begin{rem} \label{rem:new-geom}
Obviously, the geometry induced by the quasimetric $q_{\mu,s}$, and a fortiori by the metric $d_{\mu,s}$, will diverge widely from the geometry given by the original metric $d$. This is most evident in the properties of the ``new'' geodesics, in comparison with the ``old'' ones (e.g. when $X = \mathbb{R}^n$  equipped with the standard Euclidean metric and with $\mu$ being the volume element.) However, the deformation of the geometry produced by (\ref{eq:def-q}) is controlled, and many essential properties are preserved. (For further details, see \cite{Se}, \cite{Se1}.)
\end{rem}

With the risk of being a bit confusing, but to be more specific, we henceforward denote by $D_{\mu,s}$ the metric $d_{\mu,s}$ assured by Proposition \ref{prop:qm-equiv-metric}, for $q_s = q_{\mu,s}$ defined in (\ref{eq:def-q}) above.
%


\subsection{Equivalence of Nets in the Two Metrics ($d$ and $D_{\mu,s}$)}
%

The basic question for us is:

\begin{quest}
What is the relation between $\varepsilon$-nets in $d$ and in $D_{\mu,s}$?
\end{quest}

This can be decomposed into two more concrete questions:

\begin{quest}
Let $\mathcal{N}_d$ be an $d$-$\varepsilon$-net. Is it also a $D_{\mu,s}$-$\varepsilon$-net?
\end{quest}
\hspace*{-0.4cm}and the more interesting, for us

\begin{quest}
Is a $D_{\mu,s}$-$\varepsilon$-net $\mathcal{N}_{D_{\mu,s}}$ also a $d$-$\varepsilon$-net? More important, does it provide a sampling for $\mu$ as well?
\end{quest}

No general answer is available yet. However, we shall show that a positive answer exists for both questions in the important special case of {\it Ahlfors regular spaces}:

\begin{defn}
Let $(X,d,\mu)$ be a metric measure space, where $\mu$ is Borel regular.
$(X,d,\mu)$ is called {\it Ahlfors regular} (of dimension $\alpha$, or simply $\alpha$-regular) iff there exists $C_0$ and $\alpha > 0$ such that
\begin{equation} \label{eq:AhlforsReg}
\frac{1}{C_0}\cdot R^\alpha \leq \mu(B[x,R]) \leq  C_0\cdot R^\alpha \,,
\end{equation}
for any $0 < R \leq {\rm diam}X$.
\end{defn}

\begin{rem}
Sometimes, the further hypothesis that $(X,d)$ is complete is added, for convenience.
\end{rem}

\begin{rem}
For the $\alpha = 1/s$, the geometry induced by the quasimetric $q_{\mu,s}$, (or by the metric $d_{\mu,s}$) coincides, essentially, with the original geometry -- see \cite{Se}, \cite{Se1} (see also Remark \ref{rem:new-geom} above).
\end{rem}

Imposing the quite mild Ahlfors regularity condition on a metric space assures that it ``behaves in terms of size and mass distribution like Euclidean space''\footnote{\cite{Se1}, p. 15}.
Moreover, compact Riemannian manifolds, endowed with their natural volume measure, are also Ahlfors regular (see, e.g., \cite{Se0}, p. 273).

\begin{rem}
It is immediate that an Ahlfors regular space is doubling.
\end{rem}

\begin{rem}  \label{rem:atoms2}
If $(X,d,\mu)$ is Ahlfors regular, then the set $A_\mu$ of atoms of $\mu$ is countable and, if $x \in A_\mu$, then there exists $r > 0$ such that $B(x,r) = \{x\}$. (For a proof of this fact for, basically, Ahlfors $1$-regular spaces, see \cite{MS}, Theorem 1.)
\end{rem}

\begin{rem} \label{rem:doub=Haus}
It turns out that the specific measure $\mu$ in the definition above is not truly important, and in fact we can substitute for $\mu$ the Hausdorff measure $\mathcal{H}_\alpha$. Indeed, for any Borel set $E \subseteq X$ (and $\mu$ as above), there exists $C' \geq 1$ such that 
\begin{equation}
\frac{1}{C'}\cdot\mathcal{H}_\alpha \leq \mu(E) \leq C'\cdot\mathcal{H}_\alpha\,.
\end{equation}
Moreover, the Hausdorff dimension ${\rm dim}_H(X) = \alpha$.\footnote{For further results relating $\mu, D_{\mu,\alpha}$ and an arbitrary Ahlfors regular measure $\nu$ on $X$, see \cite{Se1}, p. 57 ff.}
\end{rem}


The following result is, perhaps, of no great importance by itself, it is, however, significant in our context, since it relates between the first and third sampling techniques by demarking a class of spaces for which both of the mentioned techniques are applicable:

\begin{lem}
Let $\mathcal{X} = (X,d,\mu)$ be a weak $CD(K,N)$ compact space, $K \geq 0$. Then $\mathcal{X}$ is Ahlfors $N$-regular.
\end{lem}

\begin{proof}
By \cite{Vi}, Corollary 30.12, that there exists $C = C(K,N,R)$ such that the following holds
\[\mu(B(x,r)) \geq C\cdot\mu(B(y,R))r^N\,,\]
for any $x,y,r,R$ such that $B(x,r) \subset B(y,R)$. In particular, for $x=y$ and $R = {\rm diam}(X)$, we obtain that $\mu(B(x,r)) \geq C\cdot\mu(X)r^N$, that is
\[\mu(B(x,r)) \geq C_1r^N\,,\]
where $C_1 = C\cdot\mu(X)$.

The opposite inequality
\[\mu(B(x,r)) \leq C_2r^N\]
in (\ref{eq:AhlforsReg}) follows immediately from the Bishop-Gromov inequality (\ref{thm:BG++}) for $K \geq 0$, where the constant $C_2$ is the constant that appears in the formula for the volume of the ball of radius $r$ in the space form of dimension $N$, given as a 
a function of $\int_0^rS_K^N(t)dt$.

One can obtain the precise form of the double inequality in (\ref{eq:AhlforsReg}) by choosing a constant $C^*$ as $C^* = \max\{C_2,C_1^{-1}\}$.
\end{proof}




We next prove now that in Ahlfors regular spaces, $d$-nets and $D_{\mu,s}$-nets are, indeed, equivalent. More precisely, we can formulate the following

\begin{prop} \label{prop:main}
Let $(X,d,\mu)$ be a $1/s$-regular space. Then $d$-nets and $D_{\mu,s}$-nets are equivalent.
\end{prop}

\begin{proof} In fact, we can prove the proposition above not only for the metric $D_{\mu,s}$, but also for the more lax 
quasimetric $q_{\mu,s}$. Then assertion for the metric $D_{\mu,s}$ follows from Proposition \ref{prop:qm-equiv-metric} (with different constants, of course).

\vspace*{0.2cm}
\hspace*{-0.4cm}\framebox{$d$-$\varepsilon$-nets $\Longrightarrow$ $q_{\mu,s}$-$\varepsilon$-nets}
This implication is quite easy and it holds for general $\alpha$ (not necessarily equal to $1/s$).\footnote{In fact, it holds for so called {\it normal} (metric measure) spaces, i.e. such that there exist $0 < c_1,c_2 < \infty$, satisfying
$c_1r \leq \mu(B[x,r]) \leq c_2r_2$, for all $x \in X$ and  for any $r > 0$, such that $\mu(\{x\}) < r < \mu(X)$.} 
Indeed, if $d(x,y) < \varepsilon_1$,  for some $\varepsilon_1 > 0$, then
$q_{\mu,s}(x,y) = \big(\mu\left(B[x,d(x,y)]\right) + \mu\left(B[y,d(x,y)]\right)\big)^s < \big(\mu\left(B[x,\varepsilon_1]\right) + \mu\left(B[y,\varepsilon_1]\right)\big)^s < C^\star\varepsilon_1^{s\alpha}$, where $C^\star = 2^\alpha C_0^\alpha$.
%
%
%


\vspace*{0.2cm}
\hspace*{-0.4cm}\framebox{$q_{\mu, \varepsilon}$-$\varepsilon$-nets $\Longleftarrow$ $d$-$\varepsilon$-nets}
We begin by trying to better understand (following \cite{Se}) the geometry of $q_{\mu,\varepsilon}$-balls, which we denote by $\beta$: $\beta[x,\rho] = \{y\,|\,q_{\mu,\varepsilon}(x,y) \leq r\}$. (If $\beta[x,\rho] = X$, let $\rho = \inf\{\rho'\,|\,\beta[x,\rho] = X\}$.)
%
Let $r = \inf\{r'\,|\,\beta[x,\rho] \subseteq B[x,r']\}$. Then there exists $z \in \beta[x,\rho]$ such that $d(x,z) > r/2$. It can be shown that there exists $C_3$, independent of $x,z$ and $\rho$, such that
\begin{equation}   \label{eq:BbB}
B(x,r/C_3) \subseteq \beta[x,\rho]) \subseteq B[x,r]\,.
\end{equation}

Moreover, there exists $C_4$ such that
\begin{equation}  \label{eq:rmBr}
\frac{\rho}{C_4} \leq \mu(B[x,r])^{-s} \leq C_4\rho\,.
\end{equation}

Using (\ref{eq:BbB}), (\ref{eq:rmBr}) and the doubling condition (for $\mu$), we infer that
\begin{equation}  \label{eq:rmbr}
\frac{\rho}{C_4} \leq \mu(\beta[x,\rho])^{-s} \leq C_4\rho\,.
\end{equation}

Now the remainder of the proof is quite elementary. Indeed, (\ref{eq:rmbr}) holds for $y$ instead of $x$, therefore, for $s > 1$ we have:
\[\left(\frac{\rho}{C_4}\right)^{-s} \leq \mu(\beta[x,\rho]) \leq (C_4\rho)^{-s}\,, \left(\frac{\rho}{C_4}\right)^{-s} \leq \mu(\beta[y,\rho]) \leq (C_4\rho)^{-s}\,;\]
hence
\[2\left(\frac{\rho}{C_4}\right)^{-s} \leq \mu(\beta[x,\rho]) + \mu(\beta[y,\rho]) \leq 2(C_4\rho)^{-s}\,;\]
thence
\[2^s\left(\frac{\rho}{C_4}\right) \leq \big(\mu(\beta[x,\rho]) + \mu(\beta[y,\rho])\big)^s \leq 2^s(C_4\rho)^{-s}\,;\]
that is
\[2^s\left(\frac{\rho}{C_4}\right) \leq q_{\mu,s} \leq 2^s(C_4\rho)^{-s}\,.\]
For $0 < s < 1$, the same argument applies, because, while the inequalities change orientation, the double inequality of type (\ref{eq:AhlforsReg}) still holds.

\end{proof}


\subsection{
Qui Prodest?}

The 
unavoidable question which we are confronted with is whether the sampling result above has any practical potential.
This question rises not least because the notion of Ahlfors regular spaces  it is perhaps less known to the sampling community, therefore it is natural to ask whether such spaces are fairly common or just yet another technical artifice.

We shall show that, in fact, Ahfors regular spaces are quite abundant, and also how one can construct such spaces -- in a manner that will appear, by now, quite natural.

%
%

We first have, however, to introduce yet another definition, that ensures that a metric space $(X,d)$ contains no ``isolated islands'':

\begin{defn}
A metric space (X,d) is called {\it uniformly perfect} iff there exists $C_1 > 0$ such that, for any $x \in X$ and any $0 < r \leq {\rm diam X}$, there exists $y \in X$ such that
\begin{equation}
\frac{r}{C_5} \leq d(x,y) \leq r\,.
\end{equation}
\end{defn}
(For other, equivalent definitions, see \cite{He}, \cite{Se}.)

One can show that, while the basic doubling condition ensures that balls do not grow too fast, in uniformly perfect spaces they also  
 do not decrease at a too steep rate, more precisely that 
\begin{equation} \label{eq:anti-doub1}
\mu(B[x,a^kr]) \leq (1-a)^k\mu(B[x,r])\,,
\end{equation}
for any $k \in \mathbb{N}$, any $x \in X$ and any $0 < r \leq {\rm diam}X$.

Using this fact, it is easy to prove (see \cite{Se}, \cite{He}) the following result, that shows how to canonically construct Ahlfors regular spaces:

\begin{prop} \label{prop:up+d=>Areg}
Let $(X,d)$ a uniformly perfect metric space and let $\mu$ be a doubling measure on $X$. Then $(X,d_{\mu,s},\mu)$ is $\frac{1}{s}$-Ahlfors regular for any $s > 0$. Moreover, there exists $s_0 > 0$, such that the canonical injection $(X,d) \hookrightarrow (X,d_{\mu,s})$ is quasisymmetric for any $0 < s \leq s_0$.
\end{prop}

In particular, graphs endowed with $d_{\mu,s}$ metrics are $\frac{1}{s}$-Ahlfors regular
(for any $s > 0$), hence $\varepsilon$-$d$-nets in such graphs are also $\varepsilon$-$d_{\mu,s}$-nets.

\begin{rem}
The proof of the first assertion of the proposition above is basically included in the second part of the proof of Proposition \ref{prop:main}.
\end{rem}

\begin{rem} \label{rem:atoms3}
As a consequence of (\ref{eq:anti-doub1}), it follows that doubling measures on uniformly perfect spaces have no atoms. (Compare to Remarks  \ref{rem:atoms1} and \ref{rem:atoms2}  above.)
\end{rem}


In the opposite direction, it is also easy to prove that Ahlfors regular spaces are uniformly perfect (see \cite{Se}). In fact, we have the following result (for a proof, see \cite{Se}):

\begin{prop} \label{cor:AhReg-eq-UP+dmeas}
A metric space $(X,d)$ is quasisymmetrically equivalent to an Ahlfors regular space iff $(X,d)$ is uniformly perfect and admits a doubling measure.
\end{prop}

So Ahfors regular spaces represent controlled (quasisymmetric) deformations of ``nice'' doubling spaces, where ``nicety'' is formulated in terms of the rather mild restriction given by uniform perfectness. However, the second condition in Proposition \ref{cor:AhReg-eq-UP+dmeas} above, on the existence of the doubling measure, may be difficult to verify. Fortunately, a simpler, purely metric condition is equivalent to it, at least for complete spaces, namely:

\begin{defn}
A metric space $(X,d)$ is called {\em doubling} iff there exists $D_1 \geq 1$, such that any ball in $X$, of radius $r$, can be covered by at most $D_1$ balls of radius $r/2$.
\end{defn}
\hspace*{-0.4cm}(Obviously, there is nothing special about balls, and the metric doubling condition can 
be formulated in terms of general sets of bounded diameter.)

Not surprisingly, there exists a connection between the notions of doubling metric 
and doubling measure. More precisely, we have the following

\begin{lem}
Let $(X,d)$ be a metric space such that there exists a doubling measure $\mu$ on $X$. Then $(X,d)$ is doubling (as a metric space).
\end{lem}
\hspace*{-0.4cm}(For a proof, see \cite{He}, \cite{Se}.)

The converse statement does not always true, a counterexample being provided by $(\mathbb{Q},d_{eucl})$ (see \cite{He}, p. 103).
However, it does hold for complete spaces:

\begin{thm}[Luukkainen-Saksman \cite{LS}] \label{thm:LS}
Let $(X,d)$ be a doubling, complete metric space. Then $X$ carries a doubling measure.
\end{thm}

\begin{cor}
Any compact, doubling metric spaces carries a doubling measure $\mu$.
\end{cor}

The corollary above obviously holds for finite graphs.

\begin{rem}
If $(X,d)$ is $\alpha$-Ahlfors regular, then, by Remark \ref{rem:doub=Haus} above, $\mu$ can be taken as the Hausdorff $\alpha$-measure.
\end{rem}

From the above noted equivalence between complete doubling metrics and doubling measures, and from Proposition \ref{cor:AhReg-eq-UP+dmeas}, we obtain the following

\begin{cor} \label{cor:unif-perf-equiv-Ahlfors}
Any complete, uniformly perfect metric spaces is quasisymmetrically equivalent to a Ahlfors regular space.
\end{cor}

Before we formulate the important theorem of Assouad on which we shall base our sampling result, we give, for convenience, the following definition:

\begin{defn}
If $(X,d)$ is a metric space, then the metric space $(X,d^\varepsilon), 0 < \varepsilon < 1$, is called a {\it snowflaked version} of $(X,d)$.
\end{defn}

\begin{thm}[Assouad, \cite{As1}, \cite{As2}]
Let $(X,d)$ be a doubling metric space. Then, for each $0 < \varepsilon < 1$, there exists $N$, such that its $\varepsilon$-snowflaked version is bilipschitz equivalent to a subset of $\mathbb{R}^N$, quantitatively.
\end{thm}

Here, {\it quantitatively} means that the embedding dimension $N$ and the 
bilipschitz constant 
$L$ depend solely on the doubling constant $D$ of $X$ and on the ``snowflaking'' factor $\varepsilon$, i.e.
\[N = N(D,\varepsilon),\; L = L(D,\varepsilon)\,.\]

\begin{rem}
Assouad's result does not hold, in general, for $\varepsilon = 1$. (For a counterexample, see \cite{He}, p. 99).
\end{rem}

Combining Assouad's Theorem, Corollary \ref{cor:unif-perf-equiv-Ahlfors} and our own sampling result, we can now enunciate 
the following sampling ``meta-theorem'':

\begin{thm}
\label{thm:meta-thm1}
Sampling of 
Ahlfors regular metric measure spaces is quasisymmetrically equivalent, quantitatively, 
to the sampling of sets in $\mathbb{R}^N$, for some $N$. 
\end{thm}

\begin{rem}
The beauty of Assouad's Theorem -- and even more so its applicability in the sampling of real data -- is marred by the ``course of 
dimensionality'': Given that $N = N(D,\varepsilon)$, the fear exist that, as in the case of Nash's Embedding Theorem \cite{na1}, \cite{na2}, the embedding dimension is prohibitively high for general manifolds (i.e. data). Obviously, this is even more important if low distortion -- i.e. (bi-)lipschitz constant -- is an imperative (as it usually is), that is for 
$\varepsilon$ close to $0$. And, indeed, Assouad's original construction provides $\lim_{\varepsilon \rightarrow 0}{N(D,\varepsilon)} = \infty$. So it would seem that, the price to pay for low distortion is a high embedding dimension. It is a quite recent result of Naor and Naiman \cite{NN} (itself based on ideas of Abraham, Bartal and Neiman \cite{ABN}), that, in fact, given a (separable) $D$-doubling metric space, there exist $N = N(D) \in \mathbb{N}$ and $L = L(D,\varepsilon)$, such that for any $\varepsilon \in (0,1/2)$, the $(1-\varepsilon)$-snowflaked version of $X$ admits a bilipschitz embedding in $\mathbb{R}^N$, with distortion $L$. Moreover, specific upper bounds for $N$ and $L$ are given: $N \leq a\log{D}, L \leq b\left(\frac{\log{K}}{\varepsilon}\right)^2$, where $a$ and $b$ are constants. So it appears 
that, at least as far as Assouad's Theorem is concerned, the snowflaking-based embedding is feasible.
\end{rem}

At this point, one has to ask oneself whether this result can be improved. The belief in the possibility of such an improvement rests upon the following two facts: One one hand, Assouad's Theorem assures the existence of a bilipschitz embedding, which represents a much stronger condition then mere qusysymmetry\footnote{However, quasisymmetry represents a much more flexible analytic tool, than the rigid bilipschitz condition -- see \cite{He}, \cite{Se}, \cite{Se1} for a deeper and far more detailed discussion.}. On the other hand, as we have seen, Ahlfors rigidity is not the most easy property to check directly on a metric measure space, therefore one naturally would wish to find a sampling result similar to Theorem \ref{thm:meta-thm1}, that would hold for general doubling spaces. Such a result does exist, and it makes appeal again to the quasimetric $q_{\mu,s}$ as defined by (3.3). However, we have to make an additional assumption, that ensures that $q_{\mu,s}$-lengths of curves in $\mathbb{R}^N$ do not ``shrink'' too much, due to the presence of the measure $\mu$ in the definition of $q_{\mu,s}$ (see \cite{Se}). We encode this restriction via

\begin{defn}
A doubling measure $\mu$ on $\mathbb{R}^N$ is called a {\it metric doubling measure} iff there exist a constant $C_6$, and a metric $\delta$, such that %
\[\frac{1}{C_6}\delta(x,y) \leq  q_{\mu,\frac{1}{n}} \leq C_6\delta(x,y)\,,\]
for any $x,y \in X$, where $q_{\mu,\frac{1}{n}}$ is associated to $\mu$ as in (3.3), with $s = 1/n$.
\end{defn}

We can now formulate the desired result, in terms of metric doubling measures:

\begin{thm}[Semmes \cite{Se00} Theorem 1.15, \cite{Se}, Proposition B. 20.2]
Let $(X,d)$ be a doubling metric space. Then there exists a natural number $N$ and a metric doubling measure $\mu$, such that $(X,d)$ is bilipschitz equivalent to a subset of $(\mathbb{R}^N,q_{\mu,\frac{1}{N}})$, where $q_{\mu,\frac{1}{N}}$ is as above.
\end{thm}

This is a most encouraging result, and the idea of the proof is quite simple: By Assouad's Theorem, $(M,d^{\frac{1}{2}})$ is bilipschitz equivalent to a subset $Y$ of some $\mathbb{R}^N$. The sought for measure on $\mathbb{R}^N$ will be define as $\mu = {\rm dist}(x,Y^n)dx$ -- for details of the proof see \cite{Se00}.

One would naturally would hope that $(\mathbb{R}^n,q_{\mu,\frac{1}{n}})$ can be bilipschitzly embedded in some  $\mathbb{R}^N$, for any doubling measure $\mu$. This is a quite ambitious wish and, unfortunately, it is not true in general (see \cite{Se00}). 
However, such an embedding exists for ``most'' metric doubling measures -- for a precis formulation and the proof see \cite{Se00}. Still, we can formulate the fitting sampling result (recall that given the quasimetric $q_{\mu,s}$\,, there exists a metric $d_s$ bilipschitz equivalent to it):

\begin{thm}
\label{thm:meta-thm2}
Sampling of doubling metric spaces 
is bilipschitz equivalent quantitatively 
to the sampling of sets in $(\mathbb{R}^N,d_{\frac{1}{N}})$, for some $N$, where $d_{\frac{1}{N}}$ represents the snowflaked version of $d$, associated to a certain metric doubling measure $\mu$. 
\end{thm}


\section{Discussion and Final Comments}

A number of concluding remarks, regarding the relative advantages 
of the three sampling methods exposed above, are mandatory.

As far as simplicity is concerned, then obviously the third method is the preferred one, as emphasized already (even in the title of the article):  It is the most intuitive (at least for a geometer), employing just a quite simple metric. Of course, there exists a trade-off between precision and simplicity, ensuing from the fact that, if one insists on working with an actual metric and not a ``mere'' quasimetric, then he/she has to be content with approximation provided by Proposition \ref{prop:qm-equiv-metric}.
%
%
%
Moreover, it is highly adaptable, via the parameter $s$, that allows for sampling at different scale (as envisioned originally by Semmes, for quite different ends). This is not a negligible advantage, since in many application it is not a priori clear at what scale the data should be sampled at. Imaging data is the first example that cames to mind, see, e.g. \cite{Fl}, \cite{FRKV}, \cite{K}, \cite{SHKAS}. Astronomy (in its cosmological setting) represents, probably another such case. This problem appears even more poignant for data where little information regarding the structure of the data exists, and in particular no natural dimension is available. Data from bioinformatics appertains to this category. While the second (``measure decides'') method is also relatively simple, it shares with the last approach one common weakness, that is the fact that, in the absence of of curvature, there exists no way of determining the metric  density of the sampling points, akin to that of \cite{Ca} (see also \cite{SAZ} for an application in imaging).
Therefore, at this stage, this methods seem to be feasible only for data that is intrinsically ``almost flat'' (such as it appears in \cite{DG}).

This highlights the relative advantage of the curvature-based method upon the other two: While it is, admittedly, conceptually the more complicate, and computationally quite involved, it is the only one that satisfies the density condition above and, more important, it is the only algorithmic one.


However, on a more theoretical level, the superiority belongs to the metric method, since, due to Theorem \ref{thm:meta-thm1}, it reduces -- even quantitatively -- the sampling of quite general metric measure spaces to that of subsets of $\mathbb{R}^N$. This is also important in applications, since, by employing such methods as those devised e.g. in \cite{SAZ}, more general kind of data, and not just the ``almost flat'' one can be sampled (albeit by making appeal to the extrinsic curvature of a specific embedding).


\section*{Appendix - Generalized Ricci Curvature of Metric Measure Spaces}

We bring here only the minimal amount of definitions needed as a background material for Section 2.1. (The interested reader can consult, for further details, the exhaustive monograph \cite{Vi} and, of course, the original papers \cite{LV} and \cite{St}.)

\subsection{Smooth Metric Measure Spaces}

Let $M = M^n$ be a complete, connected $n$-dimensional Riemannian manifold.
One wishes to extend results regarding Ricci curvature to the case when $M^n$ is equipped with a measure that is not $d{\rm Vol}$. 
Usually (at least in our context) such a measure is taken to be of the form
\begin{equation}
\nu(dx) = e^{-V(x)}{\rm Vol}(dx)\,,
\end{equation}
where $V:M^n \rightarrow \mathbb{R}$, $V \in \mathcal{C}^2(\mathbb{R})$. Note also that any smooth positive probability measure can be written in this manner. Then $(M,d, \nu)$, where $d$ is the geodesic distance, is a metric measure space.

\begin{rem}
A standard measure $\nu$, in the context of image processing (but not only) is the gaussian measure on $\mathbb{R}^n$:
\begin{equation}
\gamma^{(n)} = \frac{e^{-|x|^2}dx}{(2\pi)^{n/2}}\,.
\end{equation}
%
\end{rem}

To 
preserve geometric significance of the Ricci tensor, one has to modify its definition as follows:
\begin{equation} \label{eq:Ric-N-Nu}
{\rm Ric}_{N,\nu} = {\rm Ric} + \nabla^2V - \frac{\nabla V \otimes
\nabla V}{N - n}
\end{equation}
Here $\nabla V \otimes \nabla V$ is a quadratic form on $TM^n$, and $\nabla^2V$ is the Hessian matrix ${\rm Hess}$, defined as:
\begin{equation}
(\nabla V \otimes \nabla V)_x(v) = (\nabla V(x) \cdot v)^2\,.
\end{equation}
Therefore
\begin{equation}
{\rm Ric}_{N,\nu}(\dot{\gamma}) = ({\rm Ric} + \nabla^2V)(\dot{\gamma}) - \frac{(\nabla V \cdot \dot{\gamma})^2}{N - n}\,.
\end{equation}
Here $N$ is the so called {\em effective dimension} and is to be inputed.

\begin{rem}
(i) If $N < n$ then ${\rm Ric}_{N,\nu} = -\infty$

\hspace*{1.7cm} (ii) If $N = n$ then, by convention, $0 \times \infty = 0$, therefore (\ref{eq:Ric-N-Nu}) is still defined even if $\nabla V = 0$, in particular
${\rm Ric}_{n,{\rm Vol}} = {\rm Ric}$ (since, in this case $V \equiv 0$).

\hspace*{1.7cm} (iii) If $N = \infty$ then ${\rm Ric}_{\infty,\nu} = {\rm Ric} + \nabla^2V$.
\end{rem}

The Ricci curvature boundedness condition of the classical Bishop-Gromov is paralleled in the case of smooth metric measure spaces by the following immediate generalization of the classical definition:

\begin{defn}
$(M,d, \nu)$ satisfies the {\it curvature-dimension} estimate ${\rm CD}(K,N)$ iff there exist $K \in \mathbb{R}$ and $N \in [1,\infty]$, such that  ${\rm Ric}_{N,\nu} \geq K$ and $n \leq N$. (If $\nu = d{\rm Vol}$, then the first condition reduces to the classical ${\rm Ric} \geq K$.)
\end{defn}

\begin{rem}
Intuitively, ``$M$ has dimension $n$ but pretends to have dimension $N$. (Identity theft)''\footnote{J. Lott \cite{L}.}

The need for such a parametric dimension stems, in particular, 
from the desire 
to extend the Bishop-Gromov Theorem to metric spaces (or more precisely, to length spaces), for which no innate notion of dimension exists.
\end{rem}

\begin{rem}
For a number of equivalent conditions, see \cite{Vi}, Theorem 14.8.
\end{rem}



\subsection{Weak $CD(K,N)$ Spaces}

\begin{defn}
Let $(X,\mu)$ and $(Y,\nu)$ be two measure spaces. A {\it coupling} (or {\it transference (transport) plan}) of $\mu$ and $\nu$ is a measure $\pi$ on $X \times Y$ with {\it marginals} $\mu$ and $\nu$ (on $X$ and $Y$, respectively), i.e. such that, 
for all measurable sets $A \subset X$ and $B \subset Y$, the following hold: $\pi[A \times Y] = \mu[A]$ and $\pi[X \times B] = \nu[B]$.
\end{defn}

\begin{defn}
Let $(X,\mu)$ and $(Y,\nu)$ be as above and let $c = c(x,y)$ be a (positive) cost function on $X \times Y$. Consider the {\it Monge-Kantorovich minimization problem}:
\begin{equation}
\inf\int_{X \times Y}c(x,y)d\pi(x,y)\,,
\end{equation}
where the infimum is taken over all the transport plans. The transport plans attaining the infimum are called {\it optimal transport (transference) plans}.
\end{defn}

Before we can proceed, we must recall the following definition and facts:

\begin{defn}
Let $(X,d)$ be a Polish space, and let $P(X)$ denote the set of Borel probability measures on $X$.
Then the {\it Wasserstein distance} (of order $2$) on $P(X)$ is defined as
\begin{equation}
W_2(\mu,\nu) = \left(\inf\int_X{d((x,y)^2d\pi(x,y)}\right)^\frac{1}{2}\,,
\end{equation}
where the infimum is taken over all the {\it transference plans} between $\mu$ and $\nu$.

\end{defn}

\begin{defn}
The {\it Wasserstein space} $P_2(X)$ is defined as
\begin{equation}
P_2(X) = \Big\{\mu \in P(X) \,\big|\, \int_X{d(x_0,x)^2\mu(dx)} < \infty \big\}\,,
\end{equation}
where $x_0 \in X$ is an arbitrary point.

\end{defn}

\begin{rem}
The definition above does not depend upon the choice of $x_0$ and $W_2$ is a metric on $P_2(X)$. Moreover, if $X$ is Polish (compact), $P_2(X)$ is also Polish (compact).
\end{rem}


\begin{defn}
Let $(X,d)$ be a compact, geodesic, Polish space and let $\Gamma = \{\gamma:[0,1] \rightarrow X\,|\, \gamma {\rm \; a \; minimal \; geodesic}\}$, and denote by $e_t:\Gamma \rightarrow X$ the (continuous) {\em evaluation map}, $e_t:(\gamma) = \gamma(t)$. Let $E:\Gamma \rightarrow X \times X$ be defined as $E(\gamma) = (e_0(\gamma),e_1(\gamma))$. A {\it dynamical transference plan} is a pair $(\pi,\Pi)$, where $\pi$ is a transference plan and $\Pi$ is a Borel measure, such that $E_\#\Pi = \pi$. $(\pi,\Pi)$ is called {\it optimal} if $\pi$ is optimal.
\end{defn}

\begin{defn}
Let $\Pi$ be an optimal dynamical transference plan. Then the one-parameter family $\{\mu_t\}_{t \in [0,1]}, \mu_t = (e_t)_\#\Pi$ is called a {\it displacement interpolation}
\end{defn}

We can now quote the following result (\cite{Vi}, Theorem 7.21 and Corollary 7.22), connecting the geometry of the Wasserstein space to
classical mass transport:

\begin{prop}
Any displacement interpolation is a Wasserstein geodesic, and conversely, any Wasserstein geodesic is obtained as a displacement interpolation from an optimal displacement interpolation.
\end{prop}

\begin{defn}
Given $N \in [1,\infty]$, the {\it displacement convexity class} $\mathcal{DC}_N$ is defined as the set of convex, continuous functions $U:\mathbb{R}_+ \rightarrow \mathbb{R}$, $U \in \mathcal{C}^2(\mathbb{R}_+ \setminus \{0\})$, such that $U(0) = 0$ and such that
\begin{equation}
\frac{rU'(r) - U(r)}{r^{1-1/N}}
\end{equation}
is nondecreasing (as a function of $r$).
%
\end{defn}

\begin{rem}
For equivalent defining conditions for the class $\mathcal{DC}_N$ see \cite{Vi}, Definition 17.1.
\end{rem}

\begin{defn}
Let $(X,d, \nu)$ be a a locally compact metric measure space, such that the measure $\nu$ is locally finite, and let $U$ be a
continuous, convex function $U:\mathbb{R}_+ \rightarrow \mathbb{R}$, $U \in \mathcal{C}^2(\mathbb{R}_+ \setminus \{0\})$, such that $U(0) = 0$. Consider a measure  $\mu$ on $X$, having compact support, and let $\mu = \rho\nu + \mu_s$ be its Lebesgue decomposition into absolutely continuous and singular parts.

Then we define the (integral) functional $U_\nu$ (with {\it nonlinearity} $U$ and {\it reference measure} $\nu$) by
\begin{equation}
U_\nu = \int_X U\big(\rho(x)\big)\nu(dx) + U'(\infty)\mu_s[X]\,.
\end{equation}

Moreover, if $\{\pi(dy|x)\}_{x \in X}$ is a family of probability measures on $X$ and if $\beta:U \times U \rightarrow (0,\infty]$ is a measurable function, we define an (integral) functional $U_{\pi,\nu}^\beta$ (with {\it nonlinearity} $U$, {\it reference measure} $\nu$, {\it coupling} $\pi$ and {\it distortion coefficient} $\beta$) by:
\begin{equation}
U_{\pi,\nu}^\beta = \int_{U \times U} U\left(\frac{\rho(x)}{\beta(x,y)}\right)\beta(x,y)\pi(dy|x)\nu(dx) + U'(\infty)\mu_s[X]\,.
\end{equation}

\end{defn}

Usually (e.g. in the definition of weak ${\rm CD}(K,N)$ spaces) $\beta$ is taken to be the {\it reference distortion coefficients}:

\begin{defn}
Let $x,y$ be two points in a metric space $(X,d)$, and consider the numbers $K \in $, $N \in [1,\infty]$ and  $t \in [0,1]$. We define the {\it reference distortion coefficients} $\beta^{(K,N)}_{t}(x,y)$ as follows:

\begin{enumerate}

\item If $t \in (0,1]$ and $1 < N < \infty$, then
\begin{equation}
\beta^{(K,N)}_{t}(x,y) = \left\{
                  \begin{array}{ll}
                    + \infty & \mbox{if $K > 0$ and $\alpha > \pi$}\,,\\\\
                    \Big(\frac{\sin{(t\alpha)}}{t\sin{\alpha}}\Big)^{N-1} & \mbox{if $K > 0$ and $\alpha \in [0,\pi]$}\,,\\\\
                    1 & \mbox{if $K = 0$}\,,\\\\
                     \Big(\frac{\sinh{(t\alpha)}}{t\sinh{\alpha}}\Big)^{N-1} & \mbox{if $K < 0$}\,;
                  \end{array}
           \right.
\end{equation}
where
\begin{equation}
\alpha = \sqrt{\frac{|K|}{N-1}}d(x,y)\,.
\end{equation}

\item In the limit cases $N \rightarrow 1$ and $N \rightarrow \infty$, define
\begin{equation}
\beta^{(K,1)}_{t}(x,y) = \left\{
                  \begin{array}{ll}
                   +\infty & \mbox{if $K > 0$}\,,\\\\
                   1 & \mbox{if $K \leq 0$}\,;
                  \end{array}
           \right.
\end{equation}

and
\begin{equation}
\beta^{(K,\infty)}_{t}(x,y) = e^{\frac{K}{6}(1-t^2)d(x,y)\,.}
\end{equation}

\item If $t = 0$, then
\begin{equation}
\beta^{(K,N)}_{0}(x,y) = 1\,.
\end{equation}

\end{enumerate}

\end{defn}
\begin{rem}
If $X$ is the model space for ${\rm CD}(K,N)$ (see \cite{Vi} p. 387), then $\beta^{(K,N)}$ is the distortion coefficient on $X$.
\end{rem}


We can now bring the definition we are interested in:

\begin{defn}
Let $(X,d,\nu)$ be a locally compact, complete, $\sigma$-finite metric measure geodesic space, and let $K \in \mathbb{R}, N \in [1,\infty]$.
We say that $(X,d,\nu)$ satisfies a {\em weak ${\rm CD}(K,N)$ condition} (or that it is a {\em weak ${\rm CD}(K,N)$ space}) iff for any two probability measures $\mu_0, \mu_1$ with compact supports ${\rm Supp}\,\mu_1, {\rm Supp}\,\mu_2 \subset {\rm Supp}\,\nu$, there exist a {\em displacement interpolation} ${\mu_t}_{0 \leq t \leq 1}$ and an associated optimal coupling $\pi$ of $\mu_0,\mu_1$ such that, for all $U \in \mathcal{DC}_N$, and for all $t \in [0,1]$, the following holds:
\begin{equation}
U_\nu(\mu_t) \leq (1-t)\,U_{\pi,\nu}^{\beta^{(K,N)}_{1-t}}(\mu_0) + t\,U_{\tilde{\pi},\nu}^{\beta^{(K,N)}_{t}}(\mu_1)
\end{equation}
(Here we denote ${\tilde{\pi}} = S_\#\pi$, where $S(x,y) = (y,x)$.)
\end{defn}

\begin{rem}
In fact, the geodesicity condition is somewhat superfluous, since a locally compact, complete metric space is geodesic (see, e.g. \cite{He}, 9.14).
\end{rem}


\section*{Acknowledgments}
The author wishes to thank Shahar Mendelson and Gershon Wolansky for the many discussions on the sampling of metric measure spaces, that largely motivated this paper at its incipient stages.



\end{document}